\documentclass[9pt]{IEEEtran}
\usepackage{amssymb}
\usepackage{amsmath}
\usepackage{amsthm}
\usepackage{enumerate}
\usepackage{varioref}
\usepackage{psfrag,balance}
\usepackage{graphics}
\usepackage{psfrag}
\usepackage{overpic}
\usepackage{extarrows}
\usepackage{cite}
\usepackage{color}
\usepackage{leftidx}

\newcommand{\diag}{\mathrm{diag}}
\newcommand{\D}{\frak{D}}

\theoremstyle{plain}
\newtheorem{lemma}{\bf Lemma}
\newtheorem{thm}{Theorem}

\theoremstyle{remark}
\newtheorem{remark}{Remark}

\begin{document}

\title{On Convergence Rate of Leader-Following Consensus of Linear Multi-Agent Systems with Communication Noises\thanks{Part of this work has been presented at the IFAC World Congress, Capetown, South Africa, August 2014. L. Cheng, Y. Wang, Z.-G. Hou and M. Tan are with the State Key laboratory of Management and Control for Complex Systems, Institute of Automation, Chinese Academy of Sciences, Beijing 100190, China. W. Ren is with Department of Electrical and Computer Engineering, University of California, Riverside, CA 92521, USA.}\thanks{Please address all correspondences to Dr. Long Cheng at
		Email: chenglong@compsys.ia.ac.cn; Tel: 8610-82544741; Fax: 8610-82544794.}}

\author{Long~Cheng, Yunpeng~Wang, Wei~Ren, Zeng-Guang Hou, Min Tan}% <-this % stops a space

\maketitle

\begin{abstract}
This note further studies the previously proposed consensus protocol for linear multi-agent systems with communication noises in \cite{Cheng14TAC,Wang15TAC}. Each agent is allowed to have its own time-varying gain to attenuate the effect of communication noises. Therefore, the common assumption in most references that all agents have the same noise-attenuation gain is not necessary. It has been proved that if all noise-attenuation gains are infinitesimal of the same order, then the mean square leader-following consensus can be reached. Furthermore, the convergence rate of the multi-agent system has been investigated. If the noise-attenuation gains belong to a class of functions which are bounded above and below by $t^{-\beta}$ $(\beta\in(0,1))$ asymptotically, then the states of all follower agents are convergent in mean square to the leader's state with the rate characterized by a function bounded above by $t^{-\beta}$ asymptotically.
\end{abstract}


\IEEEpeerreviewmaketitle

\section{Introduction}
Communication noise is an unavoidable factor in the distributed consensus of networked multi-agent systems. It has been found in \cite{Li09Automatica} that the traditional protocol cannot solve the consensus problem with the existence of the communication noise. Therefore, how to effectively attenuate the noise effect becomes an interesting research topic. One popular idea is to employ a time-descending gain (sometimes called the stochastic-approximation type gain) in the consensus protocol. In the early study phase, many scholars investigated the stochastic consensus of first-order integral multi-agent systems with communication noises. For example, under the fixed topology case: the stochastic-approximation type gain was first employed to solve the mean square and almost sure consensus problems \cite{Huang09SIAM}; Li and Zhang proved that the mean square consensus can be achieved in the continuous-time domain if and only if the noise-attenuation gain satisfies the stochastic-approximation type condition \cite{Li09Automatica}; some necessary and sufficient conditions for ensuring the stochastic consensus with both communication noises and time delays in the discrete-time domain were presented in \cite{Liu11Automatica}. Under the switching topology case, the stochastic-approximation type protocols were also proved to be effective in both the discrete-time domain \cite{Huang10TAC, Li10TAC, Liu11Automatica} and the continuous-time domain \cite{Zhang12TAC}. It is also noted that for the continuous-time mean square leader-following consensus problem, the necessary and sufficient conditions of noise-attenuation gains can be slightly relaxed compared to the leaderless case \cite{Wang13CCDC, Wang14Automatica}. The aforementioned papers study the additive noise while the multiplicative noise has also been considered in \cite{Ni13SCL, Li14TAC, Long14IJRNC}. It is usually assumed that the noise intensity is proportional to the state differences between agents. In particular, if the state difference becomes zero, then the noise effect disappears. Therefore, the protocol for solving multiplicative noises may not require the time-descending gain. A few recent results also provided the protocols for higher-order integral multi-agent systems with communication noises. For instance, it has been proved in \cite{Cheng12TAC} that the stochastic-approximation type gain is still the necessary and sufficient condition for ensuring the mean square average consensus of second-order integral multi-agent systems under the fixed topology. Further results regarding the switching topologies were presented in \cite{Cheng13Automatica, Miao14IMA}. For the general linear multi-agent system, there are also some attempts on solving the mean square consensus and the almost sure consensus under the fixed topology \cite{Cheng14TAC} and the switching topology \cite{Wang15TAC}. However, there are still some limitations in the current study of higher-order integral multi-agent systems with communication noises. For example, all agents are required to have the same noise-attenuation gain in most existing publications, which implies that certain global information should be known by all agents; and the convergence rate of the multi-agent system is rarely considered.

Motivated by the above observation, this note first modifies the consensus protocol for general continuous-time linear multi-agent systems proposed in \cite{Cheng14TAC}. In the modified protocol, each agent is equipped with its own noise-attenuation gain. Hence there is no need to assume that all agents have the same gain. It is proved that the mean square leader-following consensus can be reached by the modified protocol if all noise-attenuation gains are infinitesimal of the same order. It is interesting to find that for the leader-following consensus of general linear multi-agent systems, the stochastic-approximation type requirement on the noise-attenuation gains can be relaxed (the square integrable condition is not necessary). Next, the convergence rate of the multi-agent system under the modified protocol is presented. Although the state-transition matrix has been explicitly obtained in \cite{Cheng14TAC} (under the assumption that all agents have the same noise-attenuation gain), the entire dynamical behavior of the multi-agent system can be determined by calculating the solution to the governing stochastic differential equation. It is still difficult to clearly tell the convergence rate of the multi-agent system since the solution to the governing stochastic differential equation is a very complicated function of noise-attenuation gains and the graph Laplacian matrix. Therefore it is more challenging to answer what is the convergence rate of the multi-agent system with the consideration of agent-dependent gains. Fortunately, by some recent results in \cite{Tang14CDC}, if we assume that the noise-attenuation gain belongs to certain representative class of functions, then quantitatively determining the convergence rate becomes possible. In this note, if the noise-attenuation gain belongs to a class of functions bounded above and below by $t^{-\beta}$ $(\beta\in(0,1))$ asymptotically, then the states of all follower agents are convergent in mean square to the leader's state with the rate characterized by a function bounded above by $t^{-\beta}$ asymptotically. This convergence rate analysis is the main improvement of this note compared to the previous conference version \cite{Wang14IFAC}.

\noindent Following notations are used throughout this note: $\mathbb{C}$, ${\mathbb R}$, ${\mathbb N}^+$ denote the field of complex numbers, the field of real numbers and the set of positive integer numbers, respectively; $I_n$ denotes the $n$-dimensional unit matrix; $1_n =(1,\cdots,1)^T \in \mathbb{R}^n$; $0_n = (0,\cdots,0)^T \in\mathbb{R}^n$;
$\otimes$ denotes the Kronecker product;
for a given matrix $X$, $X^T$ denotes its transpose; $\|X\|_2$ denotes its 2-norm;
for a random variable/vector $x$, $E(x)$ denotes its mathematical expectation;
For any two functions $f(t)$ and $g(t)$, $f(t) = \mathcal{O}(g(t))$ represents $\limsup_{t\to\infty}|f(t)/g(t)| < \infty$; $f(t) = o(g(t))$ represents $\lim_{t\to\infty}|f(t)/g(t)| =0$; $f(t) = \Theta(g(t))$ represents $0 < \liminf_{t\to\infty}|f(t)/g(t)| \leq \limsup_{t\to\infty}|f(t)/g(t)| < \infty$;
for any $x \in {\mathbb{C}}$, $\Re(x)$ denotes its real part.

\section{Preliminaries}
Consider a multi-agent system composed of $N+1$ agents labeled from $0$ to $N$. %Agent $0$ is the leader; and agents $1$ to $N$ are followers.
The communication network among agents is modeled by a digraph $\cal G = \{\mathcal{V}_\mathcal{G}, \mathcal{E}_\mathcal{G}, \mathcal{A}_\mathcal{G}\}$ where $\mathcal{V}_\mathcal{G} = \{v_0,\cdots,v_{N}\}$ denotes the node set; $\mathcal{E}_\mathcal{G} = \{e_{ij}\} \subseteq \mathcal{V}_\mathcal{G} \times \mathcal{V}_\mathcal{G}$ denotes the edge set; and $\mathcal{A}_\mathcal{G} = [a_{ij}]\in \mathbb{R}^{(N+1)\times(N+1)}$ is the weight matrix. Here $v_i$ represents the $i$th agent; $e_{ij} = (v_j, v_i) \in \mathcal{E}_\mathcal{G}$ means that there is an available communication link from agent $j$ to agent $i$; $a_{ij} \geq 0$ is the communication quality associated with the edge $e_{ij}$ and $a_{ij} > 0$ if $e_{ij} \in \mathcal{E}_\mathcal{G}$, $a_{ij} = 0$ if $e_{ij} \notin \mathcal{E}_\mathcal{G}$. Agent $j$ is called the parent of agent $i$ if $e_{ij} \in \mathcal{E}_\mathcal{G}$. The neighbor set of agent $i$ is the set of all its parent agents, i.e., $\mathcal{N}_i = \{v_j | e_{ij} \in \mathcal{E}_\mathcal{G}\}$. The agent is called the leader if its neighbor set is empty, otherwise the agent is called the follower. The Laplacian matrix of $\mathcal{G}$ is defined as $\mathcal{L} = \mbox{diag}\big(\sum^{N}_{i=0}a_{0i},\cdots,\sum^{N}_{i=0}a_{Ni}\big)- \mathcal{A}_\mathcal{G}$. There is a directed path from node $v_{i_1}$ to node $v_{i_n}$ if there is a set of nodes $\{v_{i_2},\cdots,v_{i_{(n-1)}}\}$ such that edges $e_{i_2i_1},\cdots,e_{i_ni_{(n-1)}}$ all belong to $\mathcal{E}_\mathcal{G}$. The digraph $\cal G$ is called to have a spanning tree if there exists at least one node such that there are directed paths from this node to any other nodes in $\mathcal{G}$. In this note, it is assumed that the communication graph of the multi-agent system has a spanning tree. We assume that this multi-agent system has one leader (the leader is labeled by $0$). In other words, the multi-agent system has a leader-following structure. It is obvious that the Laplacian matrix of such a multi-agent system has the following form
\begin{equation}
{\cal L}=
\begin{bmatrix}
0 & 0_N^T\\
L_1 & L_2
\end{bmatrix}.
\end{equation}
\begin{lemma}[Theorem 2 in \cite{Wang14Automatica}]\label{lem:tree}
If the communication graph $\mathcal{G}$ of the multi-agent system has a spanning tree, then all eigenvalues of $L_2$ have positive real parts. Furthermore for any diagonal matrix $D$ with positive diagonal elements, all eigenvalues of $DL_2$ have positive real parts as well.
\end{lemma}

The dynamics of the $i$th agent is described by
\begin{equation}\label{eq:4}
\dot x_i(t)=Ax_i(t)+Bu_i(t),\;\; i=0,\cdots,N,
\end{equation}
where $x_i(t) \in \mathbb{R}^n$ is the state vector, $u_i(t) \in \mathbb{R}$ is the control input,
\begin{equation}
A=
\begin{bmatrix}
0 & 1 & \cdots & 0 \\
\vdots & \vdots & \ddots & \vdots\\
0 & 0 & \cdots & 1\\
\alpha_1 & \alpha_2& \cdots & \alpha_n
\end{bmatrix}\in{\mathbb R}^{n\times n}, \quad
B=
\begin{bmatrix}
0 \\ \vdots\\0\\1
\end{bmatrix}\in{\mathbb R}^n, \nonumber
\end{equation}
$(\alpha_1,\cdots,\alpha_n)$ are coefficients determined by the agent's essential dynamical characteristics.
It is obvious that any controllable single input single output system can be transformed into this Luenberger canonical form.

The agents exchange information via a noisy communication network. The information which the $i$th agent receives from the $j$th agent is denoted by $y_{ij}(t)=x_j(t)+\rho_{ij}\eta_{ij}(t)$
where  $\eta_{ij}(t)=(\eta_{ij1}(t),\cdots,\eta_{ijn}(t))^T\in{\mathbb R}^{n}$ is the $n$-dimensional standard white noise;
$\rho_{ij}=\diag(\rho_{ij1},\cdots,\rho_{ijn})\in{\mathbb R}^{n\times n}$ ($|\rho_{ijl}|<\infty$, $l=1,\cdots,n$) denotes the noise intensity matrix.
It is assumed that $\eta_{ijl}(t)$ ($i,j=0,1,\cdots,N$; $l=1,\cdots,n$) are all mutually independent.

The control objective is to achieve the mean square leader-following consensus. That is: design control inputs $u_i(t)$ by using agent $i$ and its neighbors' information such that $\lim_{t\to\infty}E\|x_i(t) - x_0(t)\|^2 = 0$ and $\limsup_{t\to\infty}E\|x_i(t)\|^2 \leq \infty$, $\forall i=0,1,\cdots,N$.

In this note, the leader-following consensus protocol for the $i$th agent is proposed as follows
\begin{equation}\label{eq:3}
u_i(t)\!=\!K_1x_i(t)\!+\!a_i(t)\!\!\sum\nolimits_{j\in{\cal N}_i}\!\!a_{ij}K_2(y_{ij}(t)\!-\!x_i(t)), i=0,\cdots,N,
\end{equation}
where $a_i(t)>0$ is the consensus gain for the $i$th agent; $K_1=(-\alpha_1,-\alpha_2-b_1,\cdots,-\alpha_{n}-b_{n-1})$ and $K_2=(b_1,\cdots,b_{n-1},1)$; $(b_1,\cdots,b_{n-1})$ are parameters to be determined later. It is easy to see that the proposed protocol is different from the ones proposed in \cite{Cheng14TAC,Wang15TAC} because each agent has its own consensus gain $a_i(t)$.

\section{Main Results}

Let $X_F(t)=\big(x_1^T(t),\cdots,x_N^T(t)\big)^T$. Then substituting \eqref{eq:3} into \eqref{eq:4} obtains that
\begin{multline}
\dot X_F(t)=(I_N\otimes (A+BK_1)-\frak{A}(t)L_2\otimes BK_2)X_F(t)\\-\frak{A}(t)L_1\otimes BK_2x_0(t)+\frak{A}(t)\Sigma \eta(t),\nonumber
\end{multline}
where $\frak{A}(t)=\diag(a_1(t),\cdots,a_N(t))$, $\Sigma=\diag(\Sigma_1,\cdots,\Sigma_N)$, $\Sigma_i=BK_2(\rho_{i0},\rho_{i1},\cdots,\rho_{iN})$ and $\eta(t)$ is the $nN(N+1)$-dimensional standard white noise vector composed of $\eta_{ij}(t)$, $i,j =0,1,\cdots, N$.

Let $\bar X_F(t)=X_F(t)-1_N\otimes x_0(t)$. Then
\begin{equation}
  \dot{\bar X}_F(t)=(I_N\otimes (A+BK_1)-\frak{A}(t)L_2\otimes BK_2)\bar X_F(t)+\frak{A}(t)\Sigma \eta(t).\nonumber
\end{equation}
Let $\hat X(t)=(I_N\otimes K_2) \bar X_F(t)$.
It is easy to see that $K_2(A+BK_1)=0_N^T$ and $K_2BK_2=K_2$.
Therefore,
\begin{IEEEeqnarray}{ll}
	\dot{\hat X}(t)=-\frak{A}(t)L_2\hat X(t)+(I_N\otimes K_2)\frak{A}(t)\Sigma\eta(t).\label{eq:9}
\end{IEEEeqnarray}

Let $\Phi(t,t_0)$ denote the state transition matrix of $\dot{\Xi}(t)=-\frak{A}(t)L_2\Xi(t)$.
By It\^o integral formula, the solution to \eqref{eq:9} can be written as
\begin{IEEEeqnarray}{ll}\label{eq:reducedclosed}
	 {\hat X}(t)=J_1(t,t_0)+J_2(t,t_0),
\end{IEEEeqnarray}
where $J_1(t,t_0)=\Phi(t,t_0)\hat X(t_0)$, $J_2(t,t_0)=\int_{t_0}^t\Phi(t,s)(I_N\otimes K_2)\frak{A}(s)\Sigma dW(s)$ and $W(t)$ is the $nN(N+1)$-dimensional standard Brownian motion corresponding to $\eta(t)$.

Throughout this note, the following four conditions hold.
 \labelformat{enumi}{({C}#1)}
  \renewcommand\labelenumi{({C}\arabic{enumi}):}
\begin{enumerate}
	\item\label{cond:int} $\int_{0}^\infty \bar a(t)dt=\infty$, where $\bar a(t)=\max_{i=1,\cdots,N}\{a_i(t)\}$.
	\item\label{cond:class} There exist positive constants $\mu_1\le \mu_2<\infty$, $T<\infty$ and $\beta\in(0,1)$ such that for $\forall t>T$, $\mu_1t^{-\beta}\le  a_i(t)\le \mu_2 t^{-\beta}$, $i=1,\cdots,N$.
	\item\label{cond:inf} All consensus gains $\{a_1(t),\cdots,a_N(t)\}$ are infinitesimal of the same order as time goes to infinity.
	\item\label{cond:root} All roots of the following polynomial have negative real parts
	\begin{equation}
	s^{n+1}+b_{n-1}s^{n-2}+\cdots+b_2s+b_1=0.
	\end{equation}
\end{enumerate}
Since $a_1(t),\cdots, a_N(t)$ are infinitesimal of the same order as time goes to infinity, there must exist $N$ positive constants $c_1,\cdots,c_N$ such that $\lim_{t\to\infty}{a_i(t)}\big/{\bar a(t)}=c_i$, $i=1,\cdots,N$.

\begin{thm}\label{thm:1}
	If Conditions \ref{cond:int}--\ref{cond:root} hold, then the proposed protocol defined by \eqref{eq:3} can solve the mean square leader-following consensus problem of \eqref{eq:4}. Furthermore, the convergence rate of the multi-agent system is characterized by $\|E(x_i(t)-x_0(t))\|_2=\mathcal{O}(e^{\frac{-\mu_1(\lambda_{\min}-\varepsilon)}{1-\beta}}t^{1-\beta})$ and $E\|x_i(t)-x_0(t)\|_2^2=\mathcal{O}(t^{-\beta})$ $(i=1,\cdots,N)$, where $\mu_1>0$ and $\beta\in(0,1)$ are defined in \ref{cond:class}; $\lambda_{\min}=\min\{\Re(\lambda_1),\cdots,\Re(\lambda_N)\}>0$, $\lambda_1,\cdots,\lambda_N$ are eigenvalues of $CL_2$ and $C=\diag(c_1,\cdots,c_N)$; and $\varepsilon$ is any constant in $(0,\min\{1,\lambda_{\min}\})$.
\end{thm}
\begin{proof}
Since $\cal G$ has a spanning tree, by Lemma \ref{lem:tree}, $\lambda_{\min} > 0$. It is easy to see that $\frak{A}(t)L_2=\bar a(t)CL_2 +\bar a(t) D(t)L_2$, where $D(t)=(d_1(t),\cdots,d_N(t))$ and $d_i(t)={a_i(t)}\big/{\bar a(t)}-c_i$.
Then $\lim_{t\to
\infty}D(t)$ is a zero matrix.
Therefore, by Lemmas \ref{lem:6} and \ref{lem:4}, for $\forall \varepsilon \in(0,\min\{1,\lambda_{min}\})$,  there must exist two positive constants $M_1, M_2<\infty$ such that for $\forall t>t_0$,
\begin{equation}\label{eq:41}
\|\Phi(t,t_0)\|_2\le M_1 e^{-(\lambda_{\min}-\varepsilon)\int_{t_0}^t\bar a(s)ds}\le M_2 e^{-\frac{\mu_1(\lambda_{\min}-\varepsilon)}{1-\beta}t^{1-\beta}}.
\end{equation}
Therefore
\begin{equation}\label{eq:tmp1}
\|J_1(t,t_0)\|_2=\mathcal{O}(e^{-\frac{\mu_1(\lambda_{\min}-\varepsilon)}{1-\beta}t^{1-\beta}}).
\end{equation}
Moreover, it is easy to see that $E(J_2(t,t_0))=0_N$.
Therefore, $\forall i=1,\cdots,N$,
\begin{equation}\label{eq:40}
E(\hat x_i(t))=E(K_2(x_i(t)-x_0(t)))=\mathcal{O}(e^{-\frac{\mu_1(\lambda_{\min}-\varepsilon)}{1-\beta}t^{1-\beta}}),
\end{equation}
which together with Lemma \ref{lem:7} leads to the fact that $\|E(x_i(t)-x_0(t))\|_2=\mathcal{O}(e^{\frac{-\mu_1(\lambda_{\min}-\varepsilon)}{1-\beta}}t^{1-\beta})$.

It can be calculated that
\begin{IEEEeqnarray}{ll}
	&\|E(J_2(t,t_0)J_2^T(t,t_0))\|_2\IEEEnonumber\\&=\left\|\int_{t_0}^t\Phi(t,s)(I_N\otimes K_2)\frak{A}(s)\Sigma\Sigma^T\frak{A}^T(s)(I_N^T\otimes K_2^T) \Phi^T(t,s)ds\right\|_2\nonumber\\
	&\le \|I_N\otimes K_2\|_2^2\|\Sigma\|_2^2\int_{t_0}^t\|\Phi(t,s)\|_2^2\|\frak{A}(s)\|_2^2 ds\nonumber\\
	&\le \|I_N\otimes K_2\|_2^2\|\Sigma\|_2^2M_1^2\int_{t_0}^t\bar a^2(s)e^{-2\mu_1(\lambda_{\min}-\varepsilon)\int_{s}^t\bar a(\tau)d\tau}ds,\nonumber
\end{IEEEeqnarray}
which implies that
	\begin{multline}
	E\|J_2(t,t_0)\|_2^2=\mathcal{O}(\|E(J_2(t,t_0)J_2^T(t,t_0))\|_2)\\=\mathcal{O}\left(\int_{t_0}^t\bar a^2(s)e^{-2\mu_1(\lambda_{\min}-\varepsilon)\int_{s}^t\bar a(\tau)d\tau}ds\right).
	\end{multline}

	By the same procedure of Lemma A.2 in \cite{Tang14CDC}, it can be proved that
    \begin{equation*}
    \int_{t_0}^t\bar a^2(s)e^{-2\mu_1(\lambda_{\min}-\varepsilon)\int_{s}^t\bar a(\tau)d\tau}ds=\mathcal{O}(t^{-\beta}).
    \end{equation*}
	Hence, $E\|\hat X(t)\|_2^2=\mathcal{O}(t^{-\beta})$, which indicates that $E|K_2(x_i(t)-x_0(t))|^2=\mathcal{O}(t^{-\beta})$, $i=1,\cdots,N$. By Lemma \ref{lem:7} and Condition \ref{cond:root}, it can be obtained that $E\|x_i(t)-x_0(t)\|_2^2=\mathcal{O}(t^{-\beta})$, $i=1,\cdots,N$.

Finally, by \eqref{eq:3}, the closed-loop dynamics of the leader agent is $\dot x_0(t)=(A+BK_1)x_0(t)$.	Since Condition \ref{cond:root} holds, there must exist a constant $v$ such that $\lim_{t\to\infty}x_0(t)=x^*\equiv(v,0,\cdots,0)^T\in{\mathbb R}^n$. Therefore, $E\|x_i(t)\|^2_2 < \infty$, $i=0,\cdots,N$.
\end{proof}

%Furthermore,  $E|K_2(x_i(t)-x^*)|^2=\mathcal{O}(t^{-\beta})$, which together with
%
%Since Assumption (A4) hold, there must exist two constants $v_1$ and $v_2>0$ such that $\lim_{t\to\infty}x_0(t)=x^*=(v,0,\cdots,0)^T\in{\mathbb R}^n$ and $\|x_0(t)-x^*\|=O(e^{-v_2t})$.
%It together with \eqref{eq:40} leads to the fact that
%\begin{equation}
%E(K_2(x_i(t)-x^*))=O(e^{-\mu_1(\lambda_{\min}-\varepsilon)t^{1-\beta}}).\nonumber
%\end{equation}
%By Lemma \ref{lem:7} and Assumption (A4), it is obtained that
%$\|E(x_i(t)-x^*)\|_2=O(e^{-\mu_1(\lambda_{\min}-\varepsilon)t^{1-\beta}})$.

\begin{remark}
Compared to the leaderless consensus studied in \cite{Tang14CDC}, it is interesting to see that $\beta$ in \ref{cond:class} can belong to $(0,0.5)$, which means that the square integrable condition on $a_i(t)$ ($\int^{\infty}_0a^2_i(s)ds < \infty$) is not necessary. From this point of view, the leader-following consensus seems easier to be achieved than the leaderless one.
\end{remark}

\section{Simulation Examples}
\begin{figure}
	\centering
	\psfrag{a}{\scriptsize$\|x_1(t)-x_0(t)\|_2^2$}
	\psfrag{b}{\scriptsize$\|x_2(t)-x_0(t)\|_2^2$}
	\psfrag{c}{\scriptsize$\|x_3(t)-x_0(t)\|_2^2$}
	\psfrag{d}{\scriptsize$\|x_4(t)-x_0(t)\|_2^2$}
	\psfrag{e}{\large $\frac{5}{t^{0.4}}$}
\psfrag{t}{\scriptsize Time (second)}
	\includegraphics[width=0.5\textwidth]{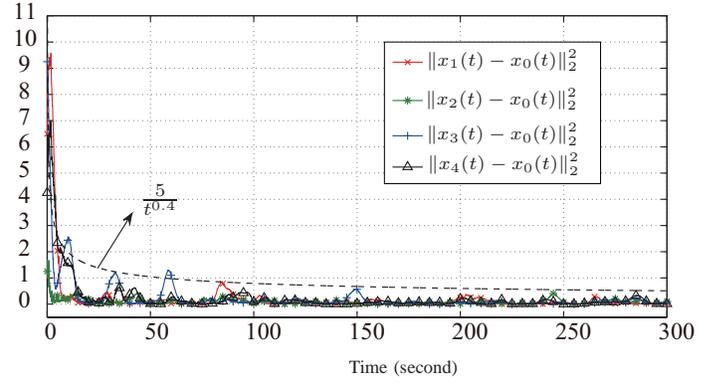}
	\caption{The trajectories of the state differences between the follower agents and the leader agent under the proposed protocol.}\label{fig:1}
\end{figure}

\begin{figure}
	\centering
	\psfrag{a}{\scriptsize$\|x_0(t)\|_2$}
	\psfrag{b}{\scriptsize$\|x_1(t)\|_2$}
	\psfrag{c}{\scriptsize$\|x_2(t)\|_2$}
	\psfrag{d}{\scriptsize$\|x_3(t)\|_2$}
	\psfrag{e}{\scriptsize$\|x_4(t)\|_2$}
\psfrag{t}{\scriptsize Time (second)}
	\includegraphics[width=0.5\textwidth]{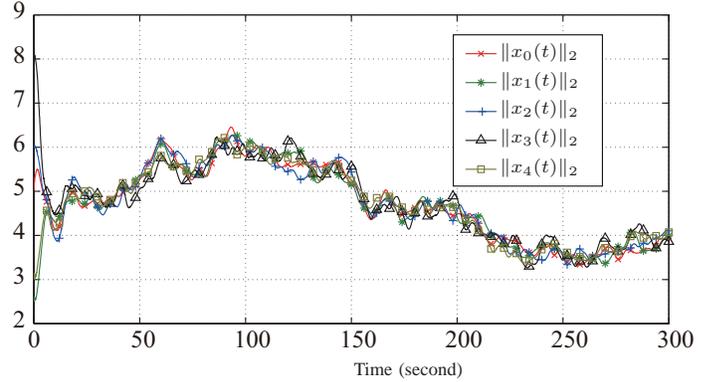}
	\caption{The trajectories of the norms of all agents' states under the leaderless case.}\label{fig:2}
\end{figure}
Consider a multi-agent system composed of five agents. In \eqref{eq:4}, $x_i(t) \in \mathbb{R}^4$, $\alpha_1=-1$, $\alpha_2=1$, $\alpha_3=0$ and $\alpha_4 = -2$. The elements of the weight matrix $\mathcal{A}_\mathcal{G}$ are set as: $a_{10} = a_{20} = a_{24} = a_{41} = a_{43} = 1$, $a_{31} = 2$ and all other elements are zero. The noise intensities $\rho_{ijl}$ in \eqref{eq:9} are all set to be $1$. The controller parameters in \eqref{eq:3} are: $K_1 = (1,-2,-3,-1)$, $K_2 = (1,3,3,1)$, $a_0(t) = 0.15/(t+1)^{0.4}$, $a_1(t) = 1.2/(t+1)^{0.4}$, $a_2(t) = 1.5/(3t+1)^{0.4}$, $a_3(t) = 0.6/(t+2)^{0.4}$ and $a_4(t) = 1.5/(4t+1)^{0.4}$. From the simulation results shown in Fig. \ref{fig:1}, it can be seen that the leader-following consensus can be achieved in the mean square sense. In addition, the trajectories of the state differences between the follower agents and the leader agent are mostly upper-bounded by $5/t^{0.4}$, which is consistent with the theoretical analysis on the convergence rate.

Next, a simulation example is conducted to show that the leaderless consensus studied in \cite{Cheng14TAC,Tang14CDC} needs the square integrable condition on the noise-attenuation gain. To this end, an extra edge $e_{04}$ is added ($a_{04}$ associated with $e_{04}$ is set to be $1$), which results in a multi-agent system without any leader. In this case, $\beta$ cannot be set to be any value belonging to $(0,0.5)$. For example, we set all noise-attenuation gains as the same value $a_i(t) = 1/(1+t)^{0.4}$, $i=0,\cdots,4$. The simulation result is given in Fig. \ref{fig:2}. By the definition of leaderless consensus in \cite{Cheng14TAC,Tang14CDC}, there must exist a random vector $x^*$ satisfying $E\|x^*\|^2 < \infty$ such that $\lim_{t\to\infty}E\|x_i(t) - x^*\|^2 = 0$, $i=0,\cdots,4$. This definition implies that $\lim_{t\to\infty}E\|x_i(t)\|^2 < \infty$, $i=0,\cdots,4$. Although by Fig. \ref{fig:2}, it seems that $\limsup_{t\to\infty}E\|x_i(t) - x_j(t)\|^2 < \infty$, the requirement that $\lim_{t\to\infty}E\|x_i(t)\|^2 < \infty$ ($i=0,\cdots,4$) cannot be satisfied.

\section{Conclusions}
This technical note relaxes the assumption made in \cite{Cheng14TAC,Wang15TAC} that all agents have the same noise-attenuation gain. Each agent is allowed to have its own time-varying gain function. It is proved that if all consensus gains are infinitesimal of the same order, then the modified protocol can still solve the mean square leader-following consensus problem of general linear multi-agent systems. In addition, this note presents the convergence rate of the multi-agent system when the noise-attenuation gains belong to a representative class of functions.

\begin{appendix}
\begin{lemma}\label{lem:6}
If Condition \ref{cond:class} holds, then $\forall b>0$, $e^{-b\int_{t_0}^t\bar a(s)ds}=\mathcal{O}(e^{\frac{-b\mu_1}{1-\beta}t^{1-\beta}})=o(t^{-\beta})$, where $\beta\in(0,1)$ is defined in Condition \ref{cond:class}.
\end{lemma}
\begin{proof}
	If Condition \ref{cond:class} holds, then $\bar a(t)>\mu_1t^{-\beta}$. Therefore,
%$ e^{-b\int_{t_0}^t\bar a(s)ds}\le e^{-b\int_{t_0}^t\mu_1s^{-\beta}ds}=e^{-\frac{b\mu_1}{1-\beta}(t^{1-\beta}-t_0^{1-\beta})}=\mathcal{O}(e^{-\frac{b\mu_1}{1-\beta}t^{1-\beta}})$.
	\begin{multline*}
    e^{-b\int_{t_0}^t\bar a(s)ds}\le e^{-b\int_{t_0}^t\mu_1s^{-\beta}ds}=e^{-\frac{b\mu_1}{1-\beta}(t^{1-\beta}-t_0^{1-\beta})}
    \\=\mathcal{O}(e^{-\frac{b\mu_1}{1-\beta}t^{1-\beta}}).
    \end{multline*}
    Let $x = t^{1-\beta}$, then $t^\beta=x^{\beta/(1-\beta)}$. Since $\beta>0$, there must exist a positive integer $n$ such that $\frac{\beta}{1-\beta}\le n$.
	By L'H\^ospotal's rule, it is obtained that
	\begin{multline}
    \!\!\lim_{t\to\infty}\frac{e^{-\frac{b\mu_1}{1-\beta}t^{1-\beta}}}{t^{-\beta}}=\lim_{x\to\infty}\frac{e^{-\frac{b\mu_1}{1-\beta}x}}{x^{-\frac{\beta}{1-\beta}}}=\lim_{x\to\infty}\frac{x^{\frac{\beta}{1-\beta}}}{e^{\frac{b\mu_1}{1-\beta}x}}\le \lim_{x\to\infty}\frac{x^n}{e^{\frac{b\mu_1}{1-\beta}x}}\\= \lim_{x\to\infty}\frac{nx^{n-1}}{\frac{b\mu_1}{1-\beta}e^{\frac{b\mu_1}{1-\beta}x}}=\cdots=\lim_{x\to\infty}\frac{n!}{(\frac{b\mu_1}{1-\beta})^ne^{\frac{b\mu_1}{1-\beta}x}}=0.\nonumber
	\end{multline}
	Hence, $e^{-b\int_{t_0}^t\bar a(s)ds}=\mathcal{O}(e^{-\frac{b\mu_1}{1-\beta}t^{1-\beta}})={o}(t^{-\beta})$.
\end{proof}

\begin{lemma}\label{lem:3}
	Consider the following differential equation
	\begin{equation}\label{eq:1}
	\dot x(t)=-\lambda a(t)(1-b(t)) x(t),
	\end{equation}
	where $a(t)\ge 0$, $\lim_{t\to\infty} b(t)=0$ and $\Re(\lambda)>0$ ($\lambda \in \mathbb{C}$). For $\forall \varepsilon>0$, there exists a positive constant $M<\infty$ such that
	$|x(t)|\le M e^{-(\Re(\lambda)-\varepsilon)\int_{t_0}^ta(s)ds}|x(t_0)|$.
\end{lemma}
\begin{proof}
	The solution to \eqref{eq:1} is $x(t)=e^{-\lambda\int_{t_0}^ta(s)(1-b(s))ds}x(t_0)$. Let $\delta=\varepsilon/\Re(\lambda)>0$. Then there must exist a finite positive constant $T \geq t_0$ such that $\forall t>T$, $1-b(t)>1-\delta$.
	\begin{itemize}
		\item If $t_0\le T\le t$, then
		\begin{IEEEeqnarray}{ll}
			|x(t)|& \le e^{-\Re(\lambda)\int_{t_0}^ta(s)(1-b(s))ds}|x(t_0)|\nonumber\\
			&\le e^{-\Re(\lambda)\int_{T}^ta(s)(1-\delta)ds}e^{-\Re(\lambda)\int_{t_0}^Ta(s)(1-\delta+\delta-b(s))ds}|x(t_0)|\nonumber\\
			&=M_1e^{-\Re(\lambda)\int_{t_0}^ta(s)(1-\delta)ds}|x(t_0)|,\nonumber
		\end{IEEEeqnarray}
		where $M_1=e^{\Re(\lambda)\int_{t_0}^Ta(s)(b(s)-\delta)ds}<\infty$.
		\item If $t_0<t<T$, then
		\begin{IEEEeqnarray}{ll}
			|x(t)|& \le e^{-\Re(\lambda)\int_{t_0}^ta(s)(1-b(s))ds}|x(t_0)|\nonumber\\
			&\le e^{-\Re(\lambda)\int_{t_0}^ta(s)(1-\delta+\delta-b(s))ds}|x(t_0)|\nonumber\\
			&=M_2e^{-\Re(\lambda)\int_{t_0}^ta(s)(1-\delta)ds}|x(t_0)|,\nonumber
		\end{IEEEeqnarray}
		where  $M_2=\sup_{t\in(t_0,T)}e^{\Re(\lambda)\int_{t_0}^ta(s)(b(s)-\delta)ds}<\infty$.
	\end{itemize}
	
	Let $M=\max\{M_1,M_2\}$. Then it is proved that $x(t)\le M e^{-(1-\delta)\Re(\lambda)\int_{t_0}^ta(s)ds}|x(t_0)|=M e^{-(\Re(\lambda)-\varepsilon)\int_{t_0}^ta(s)ds}|x(t_0)|$.
\end{proof}

\begin{lemma}[Lemma 2 in \cite{Wang13CCDC}]
	Consider the following differential equation
	\begin{equation}\label{eq:5}
	\dot x(t)=-a(t)
	\begin{bmatrix}
	\lambda & 1 &  &  \\
	&  \ddots & \ddots &\\
	& &\ddots & 1\\
	& & & \lambda
	\end{bmatrix}_{r\times r}
	\!\!\!\!\!\!\!\!\!x(t),
	\end{equation}
	where $x(t)=(x_1(t),\cdots,x_r(t))^T\in{\mathbb C}^r$.
	The state transition matrix of \eqref{eq:5} is
	\begin{equation}\label{eq:6}
	\Phi_\lambda(t,t_0)=
	\begin{bmatrix}
	P_0^\lambda(t,t_0) & P_1^\lambda(t,t_0) & \cdots & P_{r-1}^\lambda(t,t_0)\\
	0 & P_0^\lambda(t,t_0) & \cdots & P_{r-2}^\lambda(t,t_0)\\
	\vdots & \vdots & \ddots & \vdots \\
	0 & 0 & \cdots& P_0^\lambda(t,t_0)
	\end{bmatrix}
	\end{equation}
	where $\lambda\in{\mathbb C}$, $P_0^\lambda(t,t_0)=e^{-\lambda\int_{t_0}^ta(\tau)d\tau}$ and  $P_i^\lambda(t,t_0)=-\int_{t_0}^ta(\tau)P_{i-1}^\lambda(\tau,t_0)P_0^\lambda(t,\tau)d\tau$, $i=1,2,\cdots,r-1$.
\end{lemma}

\begin{lemma}\label{lem:1}
	For $\forall \varepsilon >0$, $\lambda\in{\mathbb C}$, $\Re(\lambda)>0$ there exists a positive constant  $M<\infty$ such that $\|\Phi_\lambda(t,t_0)\|_2\le M e^{-(\Re(\lambda)-\varepsilon) \int_{t_0}^ta(s)ds}$, where $\Phi_\lambda(t,t_0)$ is defined  by \eqref{eq:6}.
\end{lemma}
\begin{proof}
	It is easy to see that $|P_0^\lambda(t,t_0)|=e^{-\Re(\lambda)\int_{t_0}^ta(s)ds}$.
	Assume that $|P_i^\lambda(t,t_0)|\le M_i e^{-(\Re(\lambda)-\frac{i}{r}\varepsilon)\int_{t_0}^ta(s)ds}$, where $M_i$ is a finite positive constant.
	Then
    \begin{multline*}
      |P_{i+1}^\lambda(t,t_0)|\le\int_{t_0}^ta(\tau)|P_{i}^\lambda(\tau,t_0)||P_0^\lambda(t,\tau)|d\tau\\
      \le M_i e^{-(\Re(\lambda)-\frac{i}{r}\varepsilon)\int_{t_0}^ta(s)ds}\int_{t_0}^ta(\tau)d\tau.
    \end{multline*}

	Let $M_{i+1}={M_ir}\big/{\varepsilon e}$, where $e$ is the Euler's number. Then it is calculated that
	\begin{IEEEeqnarray}{ll}
		&|P_{i+1}^\lambda(t,t_0)|-M_{i+1}e^{-(\Re(\lambda)-\frac{i+1}{r}\varepsilon)\int_{t_0}^ta(s)ds}\nonumber\\&\le \frac{M_i\int_{t_0}^ta(\tau)d\tau}{e^{(\Re(\lambda)-\frac{i}{r}\varepsilon)\int_{t_0}^ta(s)ds}}-\frac{M_{i+1}e^{\frac{1}{r}\varepsilon\int_{t_0}^ta(s)ds}}{e^{(\Re(\lambda)-\frac{i}{r}\varepsilon)\int_{t_0}^ta(s)ds}}\nonumber\\
		&=\frac{M_i\int_{t_0}^ta(\tau)d\tau-M_{i+1}e^{\frac{1}{r}\varepsilon\int_{t_0}^ta(s)ds}}{e^{(\Re(\lambda)-\frac{i}{r}\varepsilon)\int_{t_0}^ta(s)ds}}.\nonumber
	\end{IEEEeqnarray}
	Define a function $f(\xi)=M_i\xi-M_{i+1}e^{\frac{\varepsilon}{r}\xi}$. It is easy to see that $\max_{\xi\ge0}(f(\xi))=f(\frac{r}{\varepsilon})=0$.
	Therefore, $|P_{i+1}^\lambda(t,t_0)|\le M_{i+1}e^{-(\Re(\lambda)-\frac{i+1}{r}\varepsilon)\int_{t_0}^ta(s)ds}$.
	By the mathematical induction, it can be proved that $|P_i(t,t_0)|\le \bar Me^{-(\Re(\lambda)-\varepsilon)\int_{t_0}^ta(s)ds},\; i=1,\cdots,r$,
%    \begin{equation*}
%    |P_i(t,t_0)|\le \bar Me^{-(\Re(\lambda)-\varepsilon)\int_{t_0}^ta(s)ds},\; i=1,\cdots,r,
%    \end{equation*}
	where $\bar M=\max\{1,M_1,\cdots,M_{r-1}\}$.
	Hence, there must exist a finite positive constant $M$ such that $\|\Phi_\lambda(t,t_0)\|_1\le M e^{(\Re(\lambda)-\varepsilon) \int_{t_0}^ta(s)ds}$.
\end{proof}

\begin{lemma}\label{lem:2}
	Consider the following differential equation
	\begin{equation}\label{eq:7}
	\dot x(t)=-a(t)Ax(t),
	\end{equation}
	where $a(t)\ge0$ and $x(t)\in{\mathbb R}^n$.
	The solution to this differential equation is $x(t)=\Phi(t,t_0)x(t_0)$, where $\Phi(t,t_0)$ is the state transition matrix.
	If all eigenvalues $\{\lambda_1,\cdots,\lambda_n\}$ of $A$ have positive real parts, then for $\forall \varepsilon>0$, there exists a finite positive constant $M$ such that $\|\Phi(t,t_0)\|_2\le Me^{-(\lambda_{\min}-\varepsilon)\int_{t_0}^ta(s)ds}$,
	where $\lambda_{\min}=\min\{\Re(\lambda_i)|i=1,\cdots,n\}>0$.
\end{lemma}
\begin{proof}
	There exists a transformation matrix $T$ such that $T^{-1}AT=\Lambda=\diag(\Lambda_1,\cdots,\Lambda_s)$, where $\Lambda_i\in{\mathbb C}^{r_i\times r_i}$ ($r_i\in{\mathbb N}^+$ and $\sum_{i=1}^sr_i=n$) is the Jordan block with the diagonal elements being $\lambda_i$.

	The state transition matrix $\Phi(t,t_0)$ can therefore be written in the following form
	\begin{equation}\label{eq:29}
	\Phi(t,t_0)=T^{-1}\diag(\Phi_{\lambda_1}(t,t_0),\cdots,\Phi_{\lambda_s}(t,t_0))T,
	\end{equation}
	where $\Phi_{\lambda_i}(t,t_0)$ is defined by \eqref{eq:6}.
	By Lemma \ref{lem:1}, there exists a finite positive constant $M_i$ such that
    \begin{equation*}
      \|\Phi_{\lambda_i}(t,t_0)\|_2\le M_i e^{-(\Re(\lambda_i)-\varepsilon) \int_{t_0}^ta(s)ds}\le M_i e^{-(\lambda_{\min}-\varepsilon) \int_{t_0}^ta(s)ds}.
    \end{equation*}	
	Therefore, $\|\Phi(t,t_0)\|_2\le Me^{-(\lambda_{\min}-\varepsilon) \int_{t_0}^ta(s)ds}$, where $M=\|T\|_2\|T^{-1}\|_2\max_i\{M_i\}$.
\end{proof}

\begin{lemma}\label{lem:4}
	Consider the following differential equation
	\begin{equation}\label{eq:2}
	\dot x(t)=-a(t)(A-B(t))x(t),
	\end{equation}
	where $a(t)\ge0$, $x(t)\in{\mathbb R}^n$ and $\lim_{t\to \infty}B(t)$ is a zero matrix.
	If all eigenvalues of $A$ have positive real parts, then for $\forall \varepsilon>0$, there exists a finite positive constant $M$ such that
	$\|\Psi(t,t_0)\|_2\le M_1e^{-(\lambda_{\min}-\varepsilon)\int_0^ta(s)ds}$,
	where $\Psi(t,t_0)$ is the state matrix of \eqref{eq:2} and $\lambda_{\min}$ is defined in Lemma \ref{lem:2}.
\end{lemma}
\begin{proof}
	The solution to \eqref{eq:2} can be written as
	$x(t)=\Phi(t,t_0)x(t_0)+\int_{t_0}^ta(s)\Phi(t,s)B(s)x(s)d(s)$,
	where $\Phi(t,t_0)$ is the state transition matrix of \eqref{eq:7}.
	
	By Lemma \ref{lem:2}, there exists a finite positive constant $M_1>1$ such that $\|\Phi(t,t_0)\|\le M_1e^{-(\lambda_{\min}-\varepsilon/2)\int_{t_0}^ta(s)ds}$.
	Hence,
    \begin{IEEEeqnarray}{rl}
    &\|x(t)\|_2\le\|\Phi(t,t_0)\|_2\|x(t_0)\|_2\nonumber\\
    &\quad+\int_{t_0}^ta(s)\|\Phi(t,s)\|_2\|B(s)\|_2\|x(s)\|_2d(s)\nonumber\\
	&\le M_1e^{-(\lambda_{\min}-\varepsilon/2)\int_{t_0}^ta(s)ds} \|x(t_0)\|_2\nonumber\\
    &\quad+\int_{t_0}^ta(s)e^{-(\lambda_{\min}-\varepsilon/2)\int_{s}^ta(\tau)d\tau}M_1\|B(s)\|_2\|x(s)\|_2ds.\nonumber
    \end{IEEEeqnarray}

	Consider another differential equation
	\begin{equation}\label{eq:8}
	\dot y(t)=-(\lambda_{\min}-\varepsilon/2)a(t)y(t)+a(t)M_1\|B(t)\|_2y(t),
	\end{equation}
	where $y(t) \in \mathbb{R}$. By Lemma \ref{lem:3}, there must exist a finite positive constant $M_2$ such that $\forall t>t_0$,
	\begin{equation}\label{eq:36}
	|y(t)|\le M_2e^{-(\lambda_{\min}-\varepsilon)\int_{t_0}^ta(s)ds}|y(t_0)|.
	\end{equation}
	The solution to \eqref{eq:8} is
	$y(t)=e^{-(\lambda_{\min}-\varepsilon/2)\int_{t_0}^ta(s)ds}y(t_0)+\int_{t_0}^ta(s)e^{-(\lambda_{\min}-\varepsilon/2)\int_{s}^ta(\tau)d\tau}M_1\|B(s)\|_2y(s)ds$.
    Therefore, if $y(t_0)=\|x(t_0)\|_2$, then for $\forall t\ge t_0$, $\|\Psi(t,t_0)x(t_0)\|_2=\|x(t)\|_2\le M_1y(t)\le M_1M_2e^{-(\lambda_{\min}-\varepsilon)\int_{t_0}^ta(s)ds}\|x(t_0)\|_2$.	
	By the arbitrariness of $x(t_0)$, there must exist a positive constant $M<\infty$ such that $\forall t>t_0$, $\|\Psi(t,t_0)\|_2\le Me^{-(\lambda_{\min}-\varepsilon)\int_{t_0}^ta(s)ds}$.
\end{proof}

%By the same procedure as the proof of Lemma A.1 in \cite{Tang15TAC}, the following lemma is obtained.
%\begin{lemma}
%	For $a(t)=O(t^{-\beta})$ and $b>0$, we have
%	\begin{equation}
%	e^{-b\int_{t_0}^ta(s)ds}=o(t^{-\beta}).\nonumber
%	\end{equation}
%\end{lemma}

\begin{lemma}\label{lem:7}
	Consider the following stochastic differential equation
	\begin{equation}\label{eq:37}
	\xi^{(n)}+b_{n-1}\xi^{(n-1)}\xi^{(n-1)}+\cdots+b_1\dot \xi(t)+b_0 \xi(t)=\zeta(t),
	\end{equation}
	where $\zeta(t)$ is a mean square continuous random process.
	It is assumed that $\zeta(t)$ is convergent to a random vector $\zeta^*$ in mean square, where $E\|\zeta^*\|_2^2<\infty$.
	If all roots of polynomial $s^n+b_{n-1}s^{n-1}+\cdots+b_0=0$ have negative real parts, then we have
	\begin{description}
		\item[\textbf{(I)}] $\lim_{t\to\infty}E\|\xi(t)-\zeta^*/b_0\|_2^2=0$ and $\lim_{t\to\infty}\|\xi^{(i)}(t)\|_2^2=0$, $i=1,\cdots,n$.
		\item[\textbf{(II)}] If $E(\zeta(t)-\zeta^*)=\mathcal{O}(e^{-\mu t^{\beta}})$ where $\mu>0$ and $\beta\in(0,1)$, then $E(\xi(t)-\zeta^*/b_0)=\mathcal{O}(e^{-\mu t^{\beta}})$ and $E(\xi^{(i)}(t))=\mathcal{O}(e^{-\mu t^{\beta}})$, $i=1,\cdots,n$.
		\item[\textbf{(III)}] If $E|\zeta(t)-\zeta^*|^2=\mathcal{O}(t^{-\beta})$ where $\beta\in(0,1)$, then $E|\xi(t)-\zeta^*/b_0|^2=\mathcal{O}(t^{-\beta})$ and $E|\xi^{(i)}(t)|^2=\mathcal{O}(t^{-\beta})$, $i=1,\cdots,n$.
	\end{description}
\end{lemma}
\begin{proof}	
	\textbf{(I)} see the proof of Lemma 2 in \cite{Wang15TAC}.
	
	\textbf{(II)}
	Let $\D$ denote the differential operator, namely $\D^i\xi(t)=\xi^{(i)}(t)$.
	Let $\{r_1,\cdots,r_n\}$ denote the roots of polynomial $s^n+b_{n-1}s^{n-1}+\cdots+b_0=0$, where $\Re(r_i)<0,\ i=1,\cdots,n$.
	Then the stochastic differential equation \eqref{eq:37} can be rewritten as
	$\prod_{i=1}^n(\D-r_i)\xi(t)=\zeta(t)$. Let $x_i(t)=\prod_{j=i+1}^n(\D-r_j)\xi(t)$ ($i=1,\cdots,n-1$) and $x_n(t)=\xi(t)$.
	
	{\textbf{Part I}}: By the definition of $x_1(t)$, it can be obtained that there must exist $n-2$ constants $k_1,\cdots,k_{n-2}$ such that $x_1(t)=\D^{n-1}\xi(t)+k_{n-2}\D^{n-2}\xi(t)+\cdots+k_1\D\xi(t)+\prod_{i=2}^n(-r_i)\xi(t)$.	
	Therefore, according to (I), it can be obtained that $x_1(t)$ is convergent to $\prod_{i=2}^n(-r_i)\zeta^*/b_0=-\zeta^*/{r_1}$ in mean square.
	Furthermore, $\dot x_1(t)=r_1x_1(t)+\zeta(t)$.
	Then,	
	\begin{multline}
		x_1(t)=e^{r_1(t-t_0)}x_1(t_0)+\int_{t_0}^t\!\!\!\!e^{r_1(t-s)}((\zeta(s)-\zeta^*)+\zeta^*)ds,\label{eq:16}
	\end{multline}
	which follows that $E\{x_1(t)\}=R_1(t,t_0)+R_2(t,t_0)+R_3(t,t_0)$, where $R_1(t,t_0)=e^{r_1(t-t_0)}x_1(t_0)$, $R_2(t,t_0)=\int_{t_0}^te^{r_1(t-s)}E(\zeta(s)-\zeta^*)ds$ and $R_3(t,t_0)=\int_{t_0}^te^{r_1(t-s)}E(\zeta^*)ds$.
	It is easy to see that
	\begin{equation}\label{eq:12}
	R_1(t,t_0)=\mathcal{O}(e^{r_1t}).
	\end{equation}
	
	Since $E(\zeta(t)-\zeta^*)=\mathcal{O}(e^{-\mu t^{\beta}})$, there must exist a finite positive constant  $M$ such that  $|E(\zeta(t)-\zeta^*)|\le Me^{-\mu t^{\beta}}$.
	Therefore,
	\begin{equation}\label{eq:13}
	|R_2(t,t_0)|\le M\int_{t_0}^te^{\Re(r_1)(t-s)}e^{-\mu s^{\beta}}ds.
	\end{equation}
	
	It is easy to see that $\lim_{t\to\infty}e^{-\Re(r_1)t-\mu t^{\beta}}=\infty$ and $\int_{t_0}^\infty e^{-\Re(r_1)s-\mu s^{\beta}}ds=\infty$.
	Therefore, by L'H\^ospital's rule, it is obtained that
    \begin{multline*}
    \lim_{t\to\infty}{\int_{t_0}^te^{\Re(r_1)(t-s)}e^{-\mu s^{\beta}}ds}\big/{e^{-\mu t^{\beta}}}\\
    =\lim_{t\to\infty}{\int_{t_0}^te^{-\Re(r_1)s-\mu s^{\beta}}ds}\big/{e^{-\Re(r_1)t-\mu t^{\beta}}}={-1}\big/{\Re(r_1)}.
    \end{multline*}
	Hence $\int_{t_0}^te^{\Re(r_1)(t-s)}e^{-\mu s^{\beta}}ds=\mathcal{O}(e^{-\mu t^{\beta}})$.
	This together with \eqref{eq:13} leads to the fact that
	\begin{equation}\label{eq:15}
	R_2(t,t_0)=\mathcal{O}(e^{-\mu t^{\beta}}).
	\end{equation}
	
	It can be calculated that
	\begin{equation}\label{eq:17}
	R_3(t,t_0)=\frac{E(\zeta^*)(e^{-r_1t}-e^{-r_1t_0})}{-r_1e^{-r_1t}}=\mathcal{O}(e^{r_1t}).
	\end{equation}
	
	Therefore by \eqref{eq:16}, \eqref{eq:12}, \eqref{eq:15} and \eqref{eq:17}, it is proved that
	$E(x_1(t)-{\zeta^*}/{(-r_1)})=\mathcal{O}(e^{-\mu t^{\beta}})$. By the same procedure, it can be proved that
	\begin{equation}\label{eq:18}
	E\left(x_i(t)-\zeta^*/\prod\nolimits_{j=1}^i(-r_j)\right)=\mathcal{O}(e^{-\mu t^{\beta}}),\; i=1,\cdots,n.
	\end{equation}

	{\textbf{Part II:} } It is easy to see that  $\prod_{r=1}^n(-r_i)=b_0$.
	Therefore, it is obtained from \eqref{eq:18} that
	\begin{equation}\label{eq:38}
	E(\xi(t)-\zeta^*/b_0)=\mathcal{O}(e^{-\mu t^{\beta}}).
	\end{equation}
	
	By \eqref{eq:18}, $E\big(x_{n-1}(t)+r_n\zeta^*/b_0\big)=E\big(\dot \xi(t)-r_n\xi(t)+r_n\zeta^*/b_0\big)=\mathcal{O}(e^{-\mu t^{\beta}})$, which together with \eqref{eq:38} leads to the fact that $E(\dot \xi(t))=r_nE(\xi(t)-\zeta^*/b_0)+\mathcal{O}(e^{-\mu t^{\beta}})=\mathcal{O}(e^{-\mu t^{\beta}})$.
	
	Assume that there exists a positive integer $k<n$ such that for $\forall i\in\{1,\cdots,k\}$, $E(\xi^{(i)}(t))=\mathcal{O}(e^{-\mu t^{\beta}})$.
	It is obtained from \eqref{eq:18} that $E(x_{n-k-1}(t)-\prod_{i=n-k}^n (-r_i)\zeta^*/b_0)=E(\prod_{i={n-k}}^n(\D-r_i)\xi(t)-\prod_{i=n-k}^n (-r_i)\zeta^*/b_0)=\mathcal{O}(e^{-\mu t^{\beta}})$,
	which follows that $E(\xi^{(k+1)}(t))=\mathcal{O}(e^{-\mu t^{\beta}})+\prod_{i=n-k}^n (-r_i)E(\xi(t)-\zeta^*/b_0)=\mathcal{O}(e^{-\mu t^{\beta}})$.
	By the mathematical induction, it is proved that $E(\xi^{(i)}(t))=\mathcal{O}(e^{-\mu t^{\beta}})$, $i=1,\cdots,n$.
	
	\textbf{(III)}
	By \eqref{eq:16}, it is obtained that
	\begin{equation}\label{eq:22}
		\left|x_1+\frac{\zeta^*}{r_1}\right|^2\le 3(R_4(t,t_0)+R_5(t,t_0)+R_6(t,t_0)),
	\end{equation}
	where $R_4(t,t_0)=|e^{r_1(t-t_0)}x_1(t_0)|^2$, $R_5(t,t_0)=\left|\int_{t_0}^te^{r_1(t-s)}(\zeta(s)-\zeta^*) ds\right|^2$ and $R_6(t,t_0)=\big|\int_{t_0}^te^{r_1(t-s)}\zeta^* ds+\frac{\zeta^*}{r_1}\big|^2$.
%	
%	
%	\begin{equation}\label{eq:22}
%	\left|x_1(t)-\int_{t_0}^te^{r_1(t-s)}\zeta^*ds\right|^2\le2(R_4(t,t_0)+R_5(t,t_0)),
%	\end{equation}
%	where $R_4(t,t_0)=|e^{r_1(t-t_0)}x_1(t_0)|^2$ and $R_5(t,t_0)=\left|\int_{t_0}^te^{r_1(t-s)}(\zeta(s)-\zeta^*) ds\right|^2$.
	
	It is easy to see that
	\begin{equation}\label{eq:23}
	E(R_4(t,t_0))=\mathcal{O}(e^{2\Re(r_1)t})=o(t^{-\beta}).
	\end{equation}
	According to the properties of mean square integral, it can be obtained that
	\begin{equation}\label{eq:19}
	E(R_5(t,t_0))\le\left(\int_{t_0}^te^{\Re(r_1)(t-s)}E^{\frac{1}{2}}|\zeta(s)-\zeta^*|^2 ds\right)^2.
	\end{equation}
	Since $E|\zeta(t)-\zeta^*|^2=\mathcal{O}(t^{-\beta})$, there must exist two positive constants $T$ and $M$ such that for $\forall t>T$, $E|\zeta(t)-\zeta^*|^2<Mt^{-\beta}$.

	Therefore,
	\begin{multline}
		\int_{t_0}^te^{\Re(r_1)(t-s)}E^{\frac{1}{2}}(\zeta(s)-\zeta^*)^2 ds\le e^{\Re(r_1)t}
\int_{t_0}^T\!\!\!\!e^{-\Re(r_1)s}\\\times E^{\frac{1}{2}}(\zeta(s)-\zeta^*)^2ds+M^{\frac{1}{2}}\int_{T}^t\!\!\!\!e^{\Re(r_1)(t-s)}s^{\frac{-\beta}{2}}ds.\label{eq:20}
	\end{multline}
	It is easy to see that
	\begin{equation}\label{eq:21}
	\!e^{\Re(r_1)t}\int_{t_0}^T\!\!\!\!\!\!e^{-\Re(r_1)s}E^{\frac{1}{2}}(\zeta(s)-\zeta^*)^2ds=\mathcal{O}(e^{\Re(r_1)t})=o(t^{-\beta}).
	\end{equation}
	By L'H\^ospital's rule, it is obtained that
	\begin{IEEEeqnarray}{ll}
		&\lim_{t\to\infty}\frac{\int_{T}^te^{\Re(r_1)(t-s)}s^{\frac{-\beta}{2}}ds}{t^{-\frac{\beta}{2}}}=\lim_{t\to\infty}\frac{\int_{T}^te^{-\Re(r_1)s}s^{\frac{-\beta}{2}}ds}{e^{-\Re(r_1)t}t^{-\frac{\beta}{2}}}\nonumber\\
		&=\lim_{t\to\infty}\frac{e^{-\Re(r_1)t}t^{\frac{-\beta}{2}}}{-\Re(r_1)e^{-\Re(r_1)t}t^{-\frac{\beta}{2}}-\frac{\beta}{2}e^{-\Re(r_1)t}t^{-\frac{\beta}{2}-1}}=-\frac{1}{\Re(r_1)}.\nonumber
	\end{IEEEeqnarray}
	Hence $\int_{T}^te^{r_1(t-s)}s^{\frac{-\beta}{2}}ds=\mathcal{O}(t^{-\frac{\beta}{2}})$, which together with \eqref{eq:19}, \eqref{eq:20} and \eqref{eq:21} leads to
	\begin{equation}\label{eq:24}
	E(R_5(t,t_0))=\mathcal{O}(t^{-\beta}).
	\end{equation}
	
%	By \eqref{eq:22}, \eqref{eq:23} and \eqref{eq:24}, it is obtained that $E\big|x_1(t)-\int_{t_0}^te^{r_1(t-s)}\zeta^*ds\big|^2=\mathcal{O}(t^{-\beta})$.
	It can be calculated that
	\begin{multline}\label{eq:26}
		E(R_6(t,t_0)) = E\left|\frac{e^{r_1t}}{-r_1}(e^{-r_1t}-e^{r_1t_0})\zeta^*+\frac{\zeta^*}{r_1}\right|^2\\
		=E\left|\frac{e^{r_1(t-t_0)}}{r_1}\zeta^*\right|^2=\mathcal{O}(e^{2\Re(r_1)t})=o(t^{-\beta}).
	\end{multline}
	By \eqref{eq:22}, \eqref{eq:23}, \eqref{eq:24} and \eqref{eq:26}, it is obtained that $E\left|x_1(t)+ \zeta^*/r_1\right|^2=\mathcal{O}(t^{-\beta})$.
%	\begin{equation}
%		E\left|x_1(t)+\frac{\zeta^*}{r_1}\right|^2=\mathcal{O}(t^{-\beta}).\nonumber
%	\end{equation}
	
	By the same procedure, it can be proved that
	\begin{equation}\label{eq:25}
	E\left|x_i(t)-\frac{\zeta^*}{\prod_{j=1}^i(-r_j)}\right|^2=\mathcal{O}(t^{-\beta}),\ i=1,\cdots,n,
	\end{equation}
	which indicates that $E|\xi(t)-\zeta^*/b_0|^2=\mathcal{O}(t^{-\beta})$.
	
	Since $x_{n-1}(t)=\dot \xi(t)-r_{n}\xi(t)$, it is obtained that
	\begin{IEEEeqnarray}{ll}
		&E|\dot \xi(t)|^2= E\left|x_{n-1}(t)-\frac{-r_n\zeta^*}{b_0}+\frac{-r_n\zeta^*}{b_0}+r_n\xi(t)\right|^2\nonumber\\
		&\le 2E\left|x_{n-1}(t)-\frac{-r_n\zeta^*}{b_0}\right|^2+2|r_n|^2E\left|\frac{\zeta^*}{b_0}-\xi(t)\right|^2=\mathcal{O}(t^{-\beta}).\nonumber
	\end{IEEEeqnarray}
	
	Assume that there exists a positive integer $k<n$ such that for $\forall i\in\{1,2,\cdots,k\}$, $E|\xi^{(i)}(t)|^2=\mathcal{O}(t^{-\beta})$.
	It is obtained from \eqref{eq:25} that $E\left|x_{n-k-1}(t)-\prod_{i=n-k}^n (-r_i)\zeta^*/b_0\right|^2=E\left|\prod_{i={n-k}}^n(\D-r_i)\xi(t)-\prod_{i=n-k}^n (-r_i)\zeta^*/b_0\right|^2=\mathcal{O}(t^{-\beta})$.
	There must exist $k$ constants $\rho_1,\cdots,\rho_k$ such that
	\begin{multline*}
	x_{n-k-1}(t)=\prod\nolimits_{i={n-k}}^n(\D-r_i)\xi(t)\triangleq\xi^{(k+1)}(t)\\
	+\sum_{j=1}^k\rho_j\xi^{(j)}(t)+\prod\nolimits_{i={n-k}}^n(-r_i)\xi(t).
	\end{multline*}
	Therefore, $|\xi^{(k+1)}(t)|^2=\big|x_{n-k-1}(t)-\sum_{j=1}^k\rho_j\xi^{(j)}(t)-\prod_{i={n-k}}^n(-r_i)\xi(t)\big|^2=\big|x_{n-k-1}(t)-\prod_{i=n-k}^n (-r_i)\zeta^*/b_0 -\sum_{j=1}^k\rho_j\xi^{(j)}(t)+\prod_{i=n-k}^n (-r_i)\zeta^*/b_0-\prod_{i={n-k}}^n(-r_i)\xi(t)\big|^2$, which follows that
\begin{equation*}
  E|\xi^{(k+1)}(t)|^2\le (k+2) \Bigg(E\left|x_{n-k-1}(t)-\frac{\prod_{i=n-k}^n (-r_i)\zeta^*}{b_0}\right|^2
\end{equation*}
\begin{multline*}
  +\sum_{j=1}^k|\rho_j|^2E|\xi^{(j)}(t)|^2+\left|\prod\nolimits_{i={n-k}}^n(-r_i)\right|^2E\left|\zeta^*/b_0-\xi(t)\right|^2	\Bigg)\\
  =\mathcal{O}(t^{-\beta}).
\end{multline*}
	By the mathematical induction, it is proved that $E|\xi^{(i)}(t)|=\mathcal{O}(t^{-\beta})$, $i=1,\cdots,n$.
\end{proof}
\end{appendix}

\end{document}